\documentclass[a4paper,reqno]{amsart}
\usepackage{a4,wasysym}
\usepackage{tikz}
\usetikzlibrary{shapes.geometric, arrows.meta}

\usepackage{amsmath,graphicx, amssymb,color,ulem,bbm,geometry}
\usepackage{epsf, subfigure, verbatim}
\usepackage{float}
\usepackage{epsfig}

\newtheorem{theorem}{Theorem}

\newtheorem{proposition}[theorem]{Proposition}

\newcommand{\ud}{\mathrm{d}}

\newcommand{\be}{\begin{equation}}
\newcommand{\ee}{\end{equation}}

\newtheorem{thm}{\bf Theorem}[section]

\newcommand{\Wolb}{{\it Wolbachia} }

\begin{document}

\newgeometry{left=2.1cm, right=2.1cm, top=3.8cm, bottom=3.8cm}

\title[Dengue transmission dynamics in age-structured human populations in the presence of {\em Wolbachia}]{Dengue transmission dynamics in age-structured human populations in the presence of {\em Wolbachia}}

\author[J\'{o}zsef Z. Farkas]{J\'{o}zsef Z. Farkas}
\address{J\'{o}zsef Z. Farkas, Division of Computing Science and Mathematics, University of Stirling, Stirling, FK9 4LA, United Kingdom }
\email{jozsef.farkas@stir.ac.uk}

\author[Stephen A. Gourley]{Stephen A. Gourley}

\author[Rongsong Liu]{Rongsong Liu}
\address{Rongsong Liu, Department of Mathematics and Department of Zoology and Physiology, University of Wyoming, Laramie, WY 82071, USA}
\email{Rongsong.Liu@uwyo.edu}
\date{\today}

\begin{abstract}

According to the World Health Organization the global incidence rate of dengue infections have risen drastically in recent years. It is estimated that globally the number of new infections is in the range of $100$ to $400$ million per annum. At the same time a number of recent studies reported the existence of \Wolb strains, which inhibit dengue virus replication in mosquito species that are primary vectors for dengue transmission. In this study we focus on the impact of \Wolb on dengue transmission dynamics in an age-structured human population. We introduce a mathematical model, which takes into account age-related effects, such as age-dependent human recovery and mortality rates, as well as age-dependent vector to human dengue transmission efficacy.  We deduce an explicit formula for the basic reproduction number $\mathcal{R}_0$, which allows us to assess the impact of the above mentioned age-related effects on the local asymptotic stability of the dengue free equilibrium. The formula we deduce for $\mathcal{R}_0$ also shows the intricate relationship between human demography and the presence of a dengue inhibiting \Wolb strain.

\end{abstract}
\maketitle

\section{Introduction}

Dengue is one of the most prevalent vector borne infectious diseases in tropical and sub-tropical climates. A recent study \cite{Bhatt} estimates that around $400$ million people get infected annually. While around three quarters of the cases are considered clinically asymptomatic, dengue infection can lead to severe health complications including haemorrhagic fever, which may lead to death. Although it is possible to acquire immunity to a specific strain, there are at least $4$ different virus serotypes identified to date. The spatial heterogeneity of the disease distribution across geographic areas and local outbreaks also continue to cause major disruptions. Although some countries very recently have approved the use of vaccines (see e.g. \cite{WHO}), the majority of efforts continue to focus on predicting the severity of  local outbreaks and on the control of the mosquito vector species. This is why the reproductive parasite \Wolb became a subject of great interest to the infectious disease community. 

Evolutionary biologists on the other hand have been fascinated for long by this maternally transmitted symbiont, partly as \Wolb is also one of the most common reproductive parasites in nature, see e.g. \cite{ONeill,Werren1997}. In fact according to a statistical analysis in \cite{Hilgen}, $66\%$ of all insect species are infected with one or more \Wolb strain. It is also worthwhile to note that in natural arthropod populations infection frequency within individual species tends to be very low or very high.  \Wolb often inhibits testes and ovaries of its host, and frequently it is also present in its host's eggs. It interferes with its host's reproductive mechanism in a fascinating fashion. This allows \Wolb to successfully establish itself in a large variety of arthropod species. Well-known effects of \Wolb include cytoplasmic incompatibility and feminization of genetic males, see e.g. \cite{Hoffmann1997,Telschow2005,Telschow2005b}. Another important and well-known effect of some \Wolb infections is the inducement of parthenogenesis, see e.g.  \cite{Stouthamer1997}. To incorporate all (or most) of these mechanisms in a tractable mathematical model is both interesting and challenging. 

To this end in recent decades a substantial number of mathematical modelling approaches have been developed and applied to model different types of \Wolb infections in a variety of arthropod species, including butterflies and mosquitoes. Many of the first models took the form of discrete time matrix models describing population frequencies, see e.g.  \cite{Turelli1994,Vautrin2007}. Using frequency type models a number of researchers investigated for example the possibility of coexistence of multiple \Wolb strains, each of which exhibiting different types of the well-known reproductive mechanisms associated with \Wolb, see e.g. \cite{Engelstadter2004,FHin2,Keeling2003,Vautrin2007}. These early results were particularly important, as \Wolb has been studied recently as a potential biological control to reduce the impact of mosquito born diseases. A successful introduction of a new \Wolb strain into a mosquito population requires a number of conditions to be met, one of which is that the new strain must have the ability to coexist with many of the possible existing strains.  In \cite{McMeniman2009} three key factors, namely, strong expression of cytoplasmic incompatibility, low fitness cost and high maternal transmission rate, were identified as drivers of a successful introduction of a new \Wolb strain into a wild {\em Aedes} population. In fact a number of studies have been published in recent years, which demonstrate the possibility of a successful introduction of specific \Wolb strains into a variety of mosquito species, such as {\it Aedes aegypti} and {\it Aedes albopictus}, which are vectors for mosquito borne diseases, such as dengue, yellow fever or West Nile virus, see e.g. \cite{Hancock, Hancock2, Hoffmann, Keeling2003}. Some of the earlier studies, e.g. \cite{FHin2,McMeniman2009,Rasgon2004} have primarily focused on the potential impact of the life-shortening effect of certain \Wolb strains on mosquito populations. More recent studies have investigated and demonstrated that specific \Wolb strains have the ability to inhibit or block dengue virus replication in a variety of mosquito species, see e.g. \cite{Blagrove,PengLu,Walker2011}.

Mathematical models to assess the potential impact of a dengue inhibiting \Wolb strain tend to be complex, as they must incorporate the well-known effects of \Wolb on the reproductive mechanisms of its host, as well as give a realistic description of dengue transmission between mosquitoes and humans and in turn dengue infection dynamics in the human population. Hughes and Britton in \cite{Hughes_2013} developed a system of differential equations by compartmentalising a female only mosquito population into $3$ compartments:  susceptible, exposed and dengue infectious, as well as dividing the human population into susceptible, infected and recovered classes. Their underlying model describing the mosquito population dynamics in the presence of \Wolb is rather simplistic, as it does not take into account some well-known effects such as male killing, which may for example lead to sex ratio distortions. Moreover, they focused on the case of complete maternal transmission as well as complete cytoplasmic incompatibility. In fact, cytoplasmic incompatibility is often not complete, and the evolution of its expression in a host is particularly important during phases when a new genetically modified cytoplasmic incompatibility  inducing strain of \Wolb is introduced into a resident mosquito population;  see e.g. \cite{Engelstadter2006} for more details on the evolution of the expression of cytoplasmic incompatibility.

More recently, in \cite{Ndii2016,Ndii2015} Ndii et al. introduced a $12$-dimensional system of ordinary differential equations, and presented some model simulations using MATLAB. In contrast to \cite{Hughes_2013} they focused on the effects of seasonal temperature variations in Queensland,   Australia, as it was hypothesized in \cite{Ritchie} that unusual temperatures had a major impact during a dengue outbreak in $2008/2009$. Their model is also more realistic from the point of view that competition effects at the aquatic stage of the mosquitoes are accounted for. They focused on a single dengue outbreak spanning $31$ weeks from November $2008$ to May  $2009$ that occurred in north-east Australia, hence human demography was not taken into account in their model.

In this work we focus on the impact of \Wolb on dengue transmission dynamics in an age-structured human population. Although our focus is also not on the underlying human population dynamics per say, it is clear that dengue induced mortality and recovery rates of humans vary by age, see e.g. \cite{Thai}. Also it is clear that age plays an important role in the infection dynamics during an outbreak, as older people have more likely to have acquired immunity against a specific serotype of dengue. Mosquito to human dengue transmission efficiency also varies by human age, for example simply because of different biting rates, see e.g. \cite{Maier,Rock} for more details. Here we show how these effects can be incorporated and studied in a mathematical model, which takes the form of an infinite dimensional dynamical system. Among other things, we deduce a formula for the basic reproduction number $\mathcal{R}_0$, which allows us to assess the impact of the above mentioned age-related effects on the local asymptotic stability of the dengue free equilibrium. The formula we deduce for $\mathcal{R}_0$ also shows the intricate relationship between human demography and the presence of a dengue inhibiting \Wolb strain.

\section{Modelling approach}

To describe the population dynamics of (adult) mosquitoes in the (possible) presence of a dengue inhibiting \Wolb strain we use a model we derived from individual mating rules recently in \cite{FGLY}. In particular, in \cite{FGLY} we derived from basic  principles the following sex-structured mosquito population model, incorporating the known mechanisms by which {\it Wolbachia} interferes with its host's reproduction. The key mechanisms are: cytoplasmic incompatibility, male killing, and reduction in reproductive output.

\begin{equation}\label{Wolbachia-sex}
\begin{aligned}
M^{\prime}(t)  & = - \mu_m M(t) + \frac{\lambda(F_{total}(t))}{N(t)} \left[M(t)F(t) + (1-\beta)(1-\tau)(M(t)F_w(t) + M_w(t)F_w(t)) + (1-q)M_w(t)F(t)\right], \\
F^{\prime}(t) & = - \mu_f F(t) + \frac{\lambda(F_{total}(t))}{N(t)} \left[M(t)F(t) + (1-\beta)(1-\tau)(M(t)F_w(t) + M_w(t)F_w(t)) + (1-q)M_w(t)F(t)\right],  \\
M^{\prime}_w(t) &= - \mu_{mw} M_w(t) + \frac{\lambda(F_{total}(t))}{N(t)} (1-\beta)\tau(1-\gamma)\left[M(t)F_w(t) + M_w(t)F_w(t)\right], \\
F^{\prime}_w(t) &= - \mu_{fw} F_w(t) + \frac{\lambda(F_{total}(t))}{N(t)} (1-\beta)\tau\left[M(t)F_w(t) + M_w(t)F_w(t)\right].
\end{aligned}
\end{equation}
The variables and parameters we used in model \eqref{Wolbachia-sex} above are as follows:

\begin{itemize}
\item $M$ stands for the number of \Wolb free male mosquitoes.
\item $F$ stands for the number of \Wolb free female mosquitoes.
\item $M_w$ stands for the number of \Wolb infected male mosquitoes.
\item $F_w$ stands for the number of \Wolb infected female mosquitoes.
\item $M_{\rm{total}} = M+M_w$ is the total number of male mosquitoes.
\item $F_{\rm total} = F + F_w$ is the total number of female mosquitoes.
\item $N=M+M_w+F+F_w=M_{\rm{total}}+F_{\rm{total}}$ is the total number of mosquitoes.
\item $\beta\in [0,1]$ quantifies the (possible) reduction of reproductive output of \Wolb infected female mosquitoes due for example to life-shortening.
\item $\tau\in [0,1]$ is the probability of maternal transmission of {\it Wolbachia}.
\item $q\in [0,1]$ quantifies the strengths of the expression of cytoplasmic incompatibility.
\item $\gamma\in [0,1]$ quantifies the (possible) expression of male killing induced by  {\em Wolbachia}.
\item $\lambda(F_{total})$ is the average egg laying rate, which depends on the total number of female mosquitoes.
\item $0<\mu_m$ denotes the per-capita mortality rate of \Wolb  free male mosquitoes.
\item $0<\mu_f$ denotes the per-capita mortality rate of \Wolb  free female mosquitoes.
\item $0<\mu_{mw}$ denotes the per-capita mortality rate of \Wolb infected male mosquitoes.
\item $0<\mu_{fw}$ denotes the per-capita mortality rate of \Wolb infected female mosquitoes.
\end{itemize}

\vspace{2mm}

Note that the question of co-existing \Wolb types inducing simultaneously for example male killing and cytoplasmic incompatibility have fascinated researchers for long, see e.g.  \cite{Engelstadter2009,Engelstadter2004}. Indeed it is known that there are \Wolb strains which exhibit sex-ratio distortion in their host, for example via male killing, and cytoplasmic incompatibility,  simultaneously, see e.g. \cite{Hurst}.  

Since our model \eqref{Wolbachia-sex} describing  \Wolb infection dynamics in a sex-structured mosquito population takes into account most of the known {\em Wolbachia} induced mechanisms, it is rather complex, and therefore we summarise it's basic mathematical properties below.
\begin{thm}\label{base-model-th}
Assume that $\lambda\,:\,\mathbb{R}_+\,\to\,\mathbb{R}_+$ is continuous, and that there exists an $\bar{x}>0$, such that for all $x>\bar{x}$ the function $\lambda$ is monotone decreasing, and that 
\begin{equation*}
\displaystyle\lim_{x\to\infty}\lambda(x)=0.  
\end{equation*}
Then, for all $(M(0),F(0),M_w(0),F_w(0))\in\mathbb{R}_+^4$, model \eqref{Wolbachia-sex} admits a unique (global) solution, which remains non-negative and bounded for all $t\ge 0$. 
\end{thm}
\begin{proof} Let us introduce the notation ${\bf w}(t):=(M(t),F(t),M_w(t),F_w(t))^T$, and rewrite model \eqref{Wolbachia-sex} in matrix form as 
\begin{equation}\label{matrix-form}
{\bf w}'(t)=G({\bf w}(t)),\quad {\bf w}(0)={\bf w}_0\in\mathbb{R}^4,
\end{equation}
where we define
\begin{equation*}
G({\bf w}):=\begin{Bmatrix}
\tilde{G}({\bf w}) & \text{if}\quad {\bf w}\neq {\bf 0} \\
 \,\,\,{\bf 0} & \text{if} \quad {\bf w}={\bf 0}
\end{Bmatrix},  \quad \text{and}\quad \tilde{G}({\bf w}):=\begin{pmatrix}
G_1({\bf w}) \\
G_2({\bf w}) \\
G_3({\bf w}) \\
G_4({\bf w})
\end{pmatrix},
\end{equation*}
where
\begin{equation}
\begin{aligned}
G_1({\bf w}(t))&=-\mu_m M(t) + \frac{\lambda(|F_{total}(t)|)}{N(t)} \left[M(t)F(t) + (1-\beta)(1-\tau)(M(t)F_w(t) + M_w(t)F_w(t)) + (1-q)M_w(t)F(t)\right], \\
G_2({\bf w}(t))&=-\mu_f F(t) + \frac{\lambda(|F_{total}(t)|)}{N(t)} \left[M(t)F(t) + (1-\beta)(1-\tau)(M(t)F_w(t) + M_w(t)F_w(t)) + (1-q)M_w(t)F(t)\right],\\
G_3({\bf w}(t))&=-\mu_{mw} M_w(t) + \frac{\lambda(|F_{total}(t)|)}{N(t)} (1-\beta)\tau(1-\gamma)\left[M(t)F_w(t) + M_w(t)F_w(t)\right],  \\
G_4({\bf w}(t))&=-\mu_{fw} F_w(t) + \frac{\lambda(|F_{total}(t)|)}{N(t)} (1-\beta)\tau\left[M(t)F_w(t) + M_w(t)F_w(t)\right].&
\end{aligned}
\end{equation}
Note that if ${\bf w}(0)={\bf 0}$ then ${\bf w}(t)\equiv {\bf 0}$ for all $t\ge 0$ (the trivial steady state of model \eqref{matrix-form}). At the same time we note that $\tilde{G}$ is locally Lipschitz continuous on (the open set) $\mathbb{R}^4\setminus \{{\bf 0}\}$, hence for any ${\bf w}(0)\in \mathbb{R}^4\setminus \{{\bf 0}\}$ a unique local solution of \eqref{matrix-form} exists by the Picard-Lindel\"{o}f theorem. 

\noindent To show that solutions exist globally, it is sufficient to show (see e.g. Corollary 2.5.3 in \cite{Pruss2010}) that there exists a constant $\omega\ge 0$, such that 
\begin{equation}\label{global-criterion}
\left\langle \tilde{G}({\bf w}),{\bf w}\right\rangle\le \omega\, ||{\bf w}||_2^2,
\end{equation}
for all ${\bf w}\in\mathbb{R}^4\setminus\{{\bf 0}\}$. Above in \eqref{global-criterion} $\langle \cdot,\cdot\rangle$ stands for the usual inner product on $\mathbb{R}^4$, and $||\cdot||_2$ for the standard Euclidean norm on $\mathbb{R}^4$. 

First we note that since $\lambda$ is continuous we have 
$$\displaystyle\sup_{x\in[0,\bar{x}]}\lambda(x)=\displaystyle\max_{x\in [0,\bar{x}]}\lambda(x)=:L<\infty,$$ 
hence $\lambda$ is bounded by $L$ on $\mathbb{R}_+$.
 Next noting that $M,M_w,F,F_w\le N$, and that we have $\beta,\gamma,q\in [0,1]$ we obtain the following estimate:
\begin{equation}
\begin{aligned}
\left\langle \tilde{G}({\bf w}),{\bf w}\right\rangle=  & G_1({\bf w})\,M+G_2({\bf w})\, F+G_3({\bf w})\, M_w+G_4({\bf w})\, F_w \\
\le & \lambda(|F+F_w|)\,2\,|M+M_w|\,|F+F_w| \\
\le & 2L\,|M+M_w|\,|F+F_w| \\
\le & 2L\,(M^2+M_w^2+F^2+F_w^2)=2L\, ||{\bf w}||_2^2,
\end{aligned}
\end{equation}
for all ${\bf w}\in\mathbb{R}^4\setminus\{{\bf 0}\}$, and therefore solutions exist globally.

Naturally we are only interested in non-negative solutions of model \eqref{Wolbachia-sex}. To this end we note that on $\mathbb{R}^4_+\setminus\{\bf 0\}$ we have
\begin{equation*}
G_1({\bf w})|_{M=0}\ge 0,\quad G_2({\bf w})|_{F=0}\ge 0, \quad G_3({\bf w})|_{M_w=0}\ge 0,\quad G_4({\bf w})|_{F_w=0}\ge 0,
\end{equation*}
and therefore positivity of solutions follows from results in \cite{Nagumo1942}.

Next we show boundedness of solutions with initial condition in $\mathbb{R}^4_+\setminus\{{\bf 0}\}$. Adding the second and fourth equation in \eqref{Wolbachia-sex}, and noticing that $\beta, q \in [0,1]$,  we  have
\begin{align}
F_{\rm{total}}^{\prime}(t) \leq & \displaystyle  -\min\{\mu_f, \mu_{fw}\}\, F_{\rm{total}}(t) + \frac{\lambda(F_{\rm{total}}(t))}{M_{\rm{total}}(t) + F_{\rm{total}}}(t)M_{\rm{total}}(t)F_{\rm{total}}(t)\nonumber \\
\le & \left[-\min\{\mu_f, \mu_{fw}\} + \lambda(F_{\rm{total}}(t))\right] F_{\rm{total}}(t).
\end{align}
Note that if $\min\{\mu_f, \mu_{fw}\}>L$ (the maximum of $\lambda$), then from above we see that $F_{\rm{total}}(t)\le F_{\rm{total}}(0),\,\,\forall\, t\ge 0$.  Otherwise, we have
$$
\limsup_{t\rightarrow \infty} F_{\rm{total}}(t) \leq \bar{F},
$$
where $\bar{F}$ is such that $\lambda(\bar{F}) = \min\{\mu_f,\mu_{fw}\}$. Note that $\bar{F}$ exists since we assumed that $\lambda$ is monotone decreasing for $x>\bar{x}$, and that it tends to zero as $x$ tends to infinity.

Since $F_{\rm{total}}(t)$ remains bounded it follows that $M_{\rm{total}}(t)$ is bounded as well, because adding the first and third equation in \eqref{Wolbachia-sex} we have
\begin{equation}\label{M-diff-ineq}
\begin{aligned}
M_{\rm{total}}^{\prime}(t) \leq & \displaystyle  -\min\{\mu_m, \mu_{mw}\}\,M_{\rm{total}}(t) + \frac{\lambda(F_{\rm{total}}(t))}{M_{\rm{total}}(t) + F_{\rm{total}}(t)}M_{\rm{total}}(t)F_{\rm{total}}(t)  \\
\leq & -\min\{\mu_m, \mu_{mw}\}\, M_{\rm{total}}(t)+ \lambda(F_{\rm{total}}(t))F_{\rm{total}}(t)  \\
 \leq & -\min\{\mu_m, \mu_{mw}\}\, M_{\rm{total}}(t) + \widetilde{F}L,
 \end{aligned}
 \end{equation}
where $\widetilde{F}:=\max\{\bar{F},F_{\rm{total}}(0)\}$ is an upper bound for $F_{\rm{total}}(t)$. From the differential inequality \eqref{M-diff-ineq}, we can conclude that $M_{\rm{total}}(t)$ is bounded, hence solutions of \eqref{Wolbachia-sex} remain bounded. 
\end{proof}

We note that the assumptions we imposed on the function $\lambda$ (the average egg laying rate of gravid female mosquitos)  in Theorem \ref{base-model-th} above are much less restrictive than those we imposed in \cite{FGLY}. In particular, our current assumptions on $\lambda$ allow us to use in our model some biologically realistic functions, for example an egg laying rate $\lambda$, giving rise to Allee effect. 

Concerning results on the existence and stability of equilibria of model \eqref{Wolbachia-sex}, we refer the interested reader to \cite{FGLY} (with a slight warning to note the different set of assumptions on $\lambda$). Here we only note that, in contrast for example to the model used in \cite{Hughes_2013}, it was shown in \cite{FGLY} that  model \eqref{Wolbachia-sex} may admit multiple coexistence steady states $(M^*,F^*,M^*_w,F^*_w)$, in some parameter regimes,  possibly giving rise to bi-stability.

Our goal here is to use model \eqref{Wolbachia-sex} as a basis to assess the effects of {\it Wolbachia} infection dynamics in a sex-structured mosquito population on dengue disease dynamics in an age-structured human population. To this end we adapt a simple SEI approach, similar to that in \cite{Hughes_2013}, and we further compartmentalise the \Wolb infected and uninfected female mosquito populations into susceptible: $F_{ws}/F_{s}$, exposed: $F_{we}/F_{e}$ and dengue infected: $F_{wi}/F_{i}$ classes. Notably, the age-structured human population is divided into susceptible: $h_s$, dengue infected: $h_i$, and recovered: $h_r$ classes. These assumptions yield the following coupled system of ordinary and partial differential equations.

\begin{equation*}
\begin{aligned}
 F_s^{\prime}(t)  = &\displaystyle \frac{\lambda(F_{\rm{total}}(t))}{N(t)} (M(t)F(t) + (1-\beta)(1-\tau)(M(t)F_w(t) + M_w(t)F_w(t)) + (1-q)M_w(t)F(t) ) \\
&   \displaystyle - \mu_f F_s(t) - \alpha_f p_{hf} F_s(t) \frac{H_i(t)}{H_{\rm{total}}(t)},\\
 F_e^{\prime}(t)   = &  \displaystyle  \alpha_f p_{hf} F_s(t) \frac{H_i(t)}{H_{\rm{total}}(t)}  -\mu_f F_e(t) - \nu_f F_e(t), \\
 F_i^{\prime}(t) =&  \displaystyle \nu_f F_e(t) - \mu_f F_i(t),\\
 F_{ws}^{\prime}(t)   =  & \displaystyle \frac{\lambda(F_{\rm{total}}(t))}{N(t)} (1-\beta)\tau(M(t)F_w(t) + M_w(t)F_w(t))  - \mu_{fw} F_{ws}(t) - \alpha_{fw} p_{hf} F_{ws}(t) \frac{H_i(t)}{H_{\rm{total}}(t)},\\
 F_{we}^{\prime}(t)   = & \displaystyle  \alpha_{fw} p_{hf} F_{ws}(t) \frac{H_i(t)}{H_{\rm{total}}(t)}  -\mu_{fw} F_{we}(t) - \nu_{fw} F_{we}(t), \\
 F_{wi}^{\prime}(t)  = & \displaystyle \nu_{fw} F_{we}(t) - \mu_{fw} F_{wi}(t),\\
  M^{\prime}(t)  = & \displaystyle \frac{\lambda(F_{\rm{total}}(t))}{N(t)} (M(t)F(t) + (1-\beta)(1-\tau)(M(t)F_w(t) + M_w(t)F_w(t)) + (1-q)M_w(t)F(t) ) - \mu_m M(t),\\
 M_w^{\prime}(t)   = & \displaystyle \frac{\lambda(F_{\rm{total}})}{N(t)} (1-\beta)\tau(1-\gamma)(M(t)F_w(t) + M_w(t)F_w(t))  - \mu_{mw} M_w(t), \\
 \end{aligned}
 \end{equation*}
 \begin{equation}\label{F}
 \begin{aligned}
 \frac{\partial h_s(t,a)}{\partial t}  + \frac{\partial h_s(t,a)}{\partial a}  = &   -\mu_{hs}(a) h_s(t,a)
  -\left(\alpha_f\alpha(a)\,p_{fh}F_i(t)\frac{h_s(t,a)}{H_{\rm total}(t)}+\alpha_{fw}\alpha(a)\,\varepsilon\,p_{fh}F_{wi}(t)\frac{h_s(t,a)}{H_{\rm total}(t)}\right) \\
   \frac{\partial h_i(t,a)}{\partial t}  +  \frac{\partial h_i(t,a)}{\partial a} = &  -\mu_{hi}(a)h_i(t,a)-\rho_h(a)h_i(t,a)   +\left(\alpha_f\alpha(a)\,p_{fh}F_i(t)\frac{h_s(t,a)}{H_{\rm total}(t)}   
   +\alpha_{fw}\alpha(a)\, \varepsilon\, p_{fh}F_{wi}(t)\frac{h_s(t,a)}{H_{\rm total}(t)}\right) \\
  \frac{\partial h_r(t,a)}{\partial t}  +  \frac{\partial h_r(t,a)}{\partial a}  = & -\mu_{hr}(a)h_r(t,a)+\rho_h(a)h_i(t,a),\\
   h_s(t,0)=& \int_0^\infty\beta_h(a)\left(h_s(t,a)+h_r(t,a)\right)\,\ud a, \\
 h_i(t,0)=& 0, \\
 h_r(t,0)=& 0,
\end{aligned}
\end{equation}
where we defined
\begin{equation}\label{H-variables}
H_{\rm total}(t):=  H_s(t)+H_i(t)+H_r(t):=\int_0^\infty h_s(t,a)\,\ud a+ \int_0^\infty h_i(t,a)\,\ud a+\int_0^\infty h_r(t,a)\,\ud a .
\end{equation}
The model parameters in \eqref{F} are as follows.
\begin{itemize}
\item $\alpha_f$ is the biting rate of female \Wolb uninfected mosquitoes.
\item $\alpha_{fw}$ is the biting rate of female \Wolb infected mosquitoes.
\item $\alpha(a)$ is the probability that a human of age $a$ is bitten by a mosquito.
\item $p_{hf}$ is the  probability of transmission of dengue from an infectious human to a dengue susceptible female mosquito.
\item $p_{fh}$ is the probability of transmission of dengue from a dengue infectious {\it Wolbachia} uninfected female mosquito to a susceptible human.
\item $\varepsilon\,p_{fh}$ is the probability of transmission of dengue from a dengue infectious {\it Wolbachia} infected female mosquito to a susceptible human. 
\item $\nu_f$ is the per-capita transition rate of dengue exposed female \Wolb uninfected mosquitoes to the infectious stage of dengue.
\item $\nu_{fw}$ is the per-capita transition rate of dengue exposed female \Wolb infected mosquitoes to the infectious stage of dengue.
\item $\mu_{hs}(a)$, $\mu_{hi}(a)$ and $\mu_{hr}(a)$ are the mortality rates  of susceptible, infectious and recovered humans of age $a$,  respectively.
\item $\beta_h(a)$ is the age-dependent fertility rate of humans. 
\item $\rho_h(a)$ is the age-dependent recovery rate of humans.
\end{itemize}

\begin{figure}
\includegraphics[height=121mm]{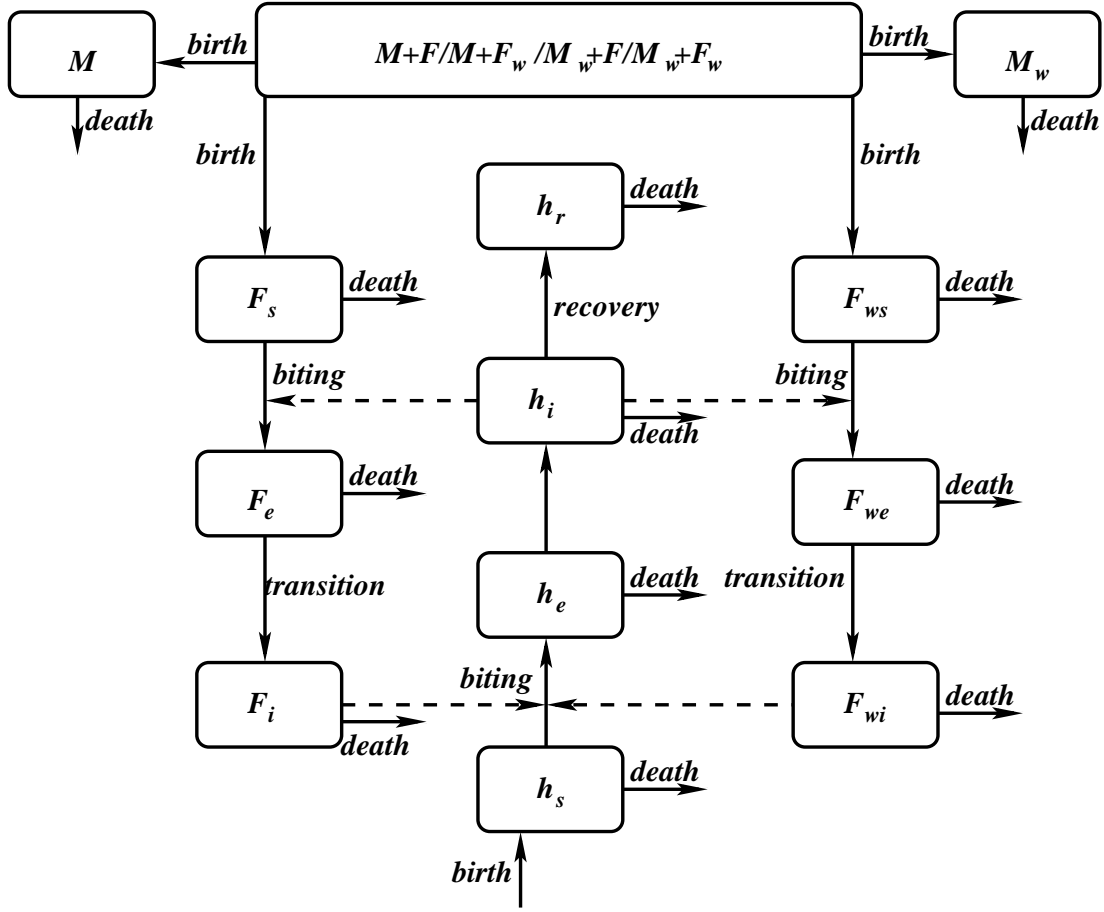}
\caption{Schematic flowchart of model \eqref{F}.}
\end{figure}

Note that we followed a simple mass action approach to model the dengue infection process, whereby in principle both \Wolb infected and uninfected female  mosquitoes may transmit dengue to humans via biting. Following recent reports in \cite{Blagrove,Hoffmann,PengLu},  we hypothesize that infection by a certain \Wolb strain blocks or at least inhibits dengue virus replication and in turn the probability of dengue transmission from \Wolb infected mosquitoes to humans. The likelihood/strength of this is measured by the parameter $\varepsilon\in[0,1]$.  It is well documented that the duration of  dengue fever is typically a few weeks, which is very short compared to the human reproductive cycle, hence we assume that humans are born susceptible. Also since the overwhelming majority of new dengue infections are caused by mosquitoes, and the event of a dengue infected mother giving birth to dengue infectious offspring is rather rare we assume that the flux of infected human offspring is zero. We also note that as we introduced an age-dependent biting rate, the term $\alpha(a)\,p_{fh}$ can be reinterpreted and used to understand and assess the impact of age-related immunity effects. 
Indeed it is well documented that it is possible to acquire immunity against one of the $4$ serotypes of dengue, although how long such immunity lasts and the possibility of cross-infection remains a major challenge, see for example \cite{Aguas} for more details. It is then clear that a portion of the older individuals have possibly acquired immunity against a specific serotype of dengue, hence the rate of new infections in the younger age groups will be often larger. This can be accounted for in the function $\alpha(a)\,p_{fh}$, and this is another reason while it is apparent that the age-structure of the human population plays a major role on dengue disease dynamics.

\section{Results}

  First note that in the absence of dengue, our model \eqref{F} decouples into the sex-structured mosquito population model \eqref{Wolbachia-sex}, and the linear age-structured model
\begin{equation}\label{PDE-steady}
\frac{\partial h_s(t,a)}{\partial t}  +  \frac{\partial h_s(t,a)}{\partial a}  =   -\mu_{hs}(a) h_s(t,a),\quad h_s(t,0)=\int_0^\infty\beta_h(a)h_s(t,a)\,\ud a.
\end{equation}
The age-structured model \eqref{PDE-steady} admits non-trivial steady states of the form
\begin{equation}
h_s^*(a)=h_s^*(0)\exp\left(-\int_0^a \mu_{hs}(\bar{a})\,\ud\bar{a}\right),
\label{040915_1}
\end{equation}
if and only if
\begin{equation}\label{PDE-steady2}
\int_0^\infty\beta_h(a)\exp\left(-\int_0^a \mu_{hs}(\bar{a})\,\ud\bar{a}\right)\ud a=1.
\end{equation}
Note that the expression on the left hand side of \eqref{PDE-steady2} is the net reproduction number of the (susceptible) human population. Since our focus here is not on the human population dynamics it is natural to assume that \eqref{PDE-steady2} holds, see also Theorem \ref{theorem2} later on. The results we recalled in the previous section then show that a variety of dengue-free steady states exist, depending on the specific values of the model parameters. Naturally, we are interested to study the stability of such equilibria.  To this end we linearise the $h_i$ equation of~(\ref{F}) at a dengue-free equilibrium, which reads
\begin{equation}
\frac{\partial h_i(t,a)}{\partial t}+\frac{\partial h_i(t,a)}{\partial a} =  -(\mu_{hi}(a)+\rho_h(a))h_i(t,a) 
 +\alpha_f\alpha(a)\,p_{fh}\frac{h_s^*(a)}{H^*}F_i(t)+\varepsilon\,\alpha_{fw}\alpha(a)\,p_{fh}\frac{h_s^*(a)}{H^*}F_{wi}(t), 
\label{040915_2}
\end{equation}
where we let $H^*$ denote the total  number of humans at the dengue-free equilibrium. The boundary condition accompanying \eqref{040915_2} simply reads $h_i(t,0)=0$, since we assumed that humans are born susceptible. If we define   $h_i^{\xi}(a):=h_i(a+\xi,a)$, then this function satisfies the equation 
\begin{equation*}
\frac{dh_i^{\xi}(a)}{da}+(\mu_{hi}(a)+\rho_h(a))h_i^{\xi}(a)=\alpha_f\alpha(a)\,p_{fh}\frac{h_s^*(a)}{H^*}F_i(a+\xi)+\varepsilon\,\alpha_{fw}\alpha(a)\,p_{fh}\frac{h_s^*(a)}{H^*}F_{wi}(a+\xi).
\end{equation*}
Since $h_i^{\xi}(0)=h_i(\xi,0)=0$, the solution of the equation above is
\begin{equation*}
h_i^{\xi}(a)=  \int_0^a\exp\left(-\int_{\bar{a}}^a(\mu_{hi}(\eta)+\rho_h(\eta))\,\ud\eta\right) 
  \left[\alpha_f\alpha(\bar{a})\,p_{fh}\frac{h_s^*(\bar{a})}{H^*}F_i(\bar{a}+\xi)+\varepsilon\,\alpha_{fw}\alpha(\bar{a})\,p_{fh}\frac{h_s^*(\bar{a})}{H^*}F_{wi}(\bar{a}+\xi)\right]\,\ud\bar{a}.
\end{equation*}
Writing  $\xi=t-a$, we have
\begin{equation*}
\begin{aligned}
h_i(t,a)=   
\int_0^a  & \exp\left(-\int_{\bar{a}}^a(\mu_{hi}(\eta)+\rho_h(\eta))\,\ud\eta\right) \\
 \times & \left[\alpha_f\alpha(\bar{a})\,p_{fh}\frac{h_s^*(\bar{a})}{H^*}F_i(\bar{a}+t-a)+\varepsilon\,\alpha_{fw}\alpha(\bar{a})\,p_{fh}\frac{h_s^*(\bar{a})}{H^*}F_{wi}(\bar{a}+t-a)\right]\,\ud\bar{a}.
 \end{aligned}
\end{equation*}
The total number of infectious humans at time $t$ is $H_i(t)=\displaystyle\int_0^{\infty}h_i(t,a)\,da$, and therefore we have
\begin{equation}
\begin{aligned}
H_i(t)= \int_0^{\infty}\int_0^a  & \exp\left(-\int_{\bar{a}}^a(\mu_{hi}(\eta)+\rho_h(\eta))\,\ud\eta\right) \\
 & \times \left[\alpha_f\alpha(\bar{a})\,p_{fh}\frac{h_s^*(\bar{a})}{H^*}F_i(\bar{a}+t-a)+\varepsilon\,\alpha_{fw}\alpha(\bar{a})\,p_{fh}\frac{h_s^*(\bar{a})}{H^*}F_{wi}(\bar{a}+t-a)\right]\,\ud\bar{a}\,\ud a.
\end{aligned}
\label{040915_3}
\end{equation}
With $H_i(t)$ given by~(\ref{040915_3}), the linearised equations for the exposed and infectious compartments at the dengue-free equilibrium are
\begin{equation}\label{eq-compartments}
\begin{aligned}
F_e'(t)  = & \alpha_f p_{hf}\frac{F_s^*}{H^*}H_i(t)-(\mu_f+\nu_f)F_e(t), \\
F_i'(t)  = & \nu_f F_e(t) - \mu_f F_i(t), \\
F_{we}'(t)  = & \alpha_{fw}p_{hf}\frac{F_{ws}^*}{H^*}H_i(t)-(\mu_{fw}+\nu_{fw})F_{we}(t), \\
F_{wi}'(t)  = & \nu_{fw}F_{we}(t)-\mu_{fw}F_{wi}(t),
\end{aligned}
\end{equation}
which is a system of distributed delay equations, due to the nature of the expression \eqref{040915_3} for the total population size of infected humans. Seeking solutions of the form
\begin{equation*}
(F_e(t),F_i(t),F_{we}(t),F_{wi}(t))=e^{\lambda t}(c_1,c_2,c_3,c_4),\quad \lambda\in\mathbb{C}
\end{equation*}
where $c_1,c_2,c_3,c_4$ are constants, and letting
\begin{equation}
\Phi(\lambda)=\int_0^{\infty}\int_0^a\exp\left(-\int_{\bar{a}}^a(\mu_{hi}(\eta)+\rho_h(\eta))\,\ud\eta\right)\alpha(\bar{a})\,h_s^*(\bar{a})e^{\lambda(\bar{a}-a)}\,\ud\bar{a}\,\ud a,
\label{040915_4}
\end{equation}
we deduce the following characteristic equation for $\lambda$:
\begin{equation}\label{040915_5}
\begin{aligned}
 \left(\lambda+\mu_f+\nu_f-\frac{\alpha_f^2p_{hf}p_{fh}F_s^*\nu_f}{H^{*^2}}\frac{\Phi(\lambda)}{\lambda+\mu_f}\right) &
\left(\lambda+\mu_{fw}+\nu_{fw}-\frac{\varepsilon\,\alpha_{fw}^2p_{hf}p_{fh}F_{ws}^*\nu_{fw}}{H^{*^2}}\frac{\Phi(\lambda)}{\lambda+\mu_{fw}}\right) \\
 & = \varepsilon\,\alpha_f^2\alpha_{fw}^2p_{hf}^2p_{fh}^2\frac{F_s^*F_{ws}^*\nu_f\nu_{fw}}{H^{*^4}}\frac{(\Phi(\lambda))^2}{(\lambda+\mu_f)(\lambda+\mu_{fw})}.
\end{aligned}
\end{equation}

A careful study of the characteristic equation \eqref{040915_5} allows us to establish the following result  concerning the local asymptotic stability of a dengue-free steady state. The proof of the theorem below is contained in Appendix A.

\begin{thm}\label{theorem1}
Suppose that a dengue-free equilibrium is locally asymptotically stable as a solution of the subsystem of system~(\ref{F}) in which the $F_e$, $F_i$, $F_{we}$, $F_{wi}$, $h_i$ and $h_r$ variables remain identically zero; and let
\begin{equation}\label{160216_1}
\begin{aligned}
\mathcal{R}_0:= & \,\,\frac{p_{hf}p_{fh}h_s^*(0)}{H^{*^2}}\left[\frac{\varepsilon\,\alpha_{fw}^2F_{ws}^*\nu_{fw}}{\mu_{fw}(\mu_{fw}+\nu_{fw})}
+\frac{\alpha_{f}^2F_{s}^*\nu_{f}}{\mu_{f}(\mu_{f}+\nu_{f})}  
\right]  \\
 & \times
\int_0^{\infty}\int_0^a\alpha(\bar{a})\,\exp\left\{-\int_0^{\bar{a}}\mu_{hs}(\eta)\,\ud\eta
-\int_{\bar{a}}^a(\mu_{hi}(\eta)+\rho_h(\eta))\,\ud\eta\right\}\,\ud\bar{a}\,\ud a. 
\end{aligned}
\end{equation}
Then, if $\mathcal{R}_0<1$ holds, the dengue-free equilibrium  is also locally asymptotically stable to perturbations involving small introductions of dengue.
\end{thm}

Note that the formula we deduced above for (the basic reproductive number) $\mathcal{R}_0$ characterizes the local stability of the dengue-free equilibrium of our mathematical model \eqref{F}.  Therefore we hypothesize that the severity of a future dengue outbreak can be estimated by parametrising our model and computing (or at least approximating) the value of $\mathcal{R}_0$ given by formula  \eqref{160216_1}. Naturally one would expect that larger values of $\mathcal{R}_0$ will potentially result in more severe dengue outbreaks.  Indeed, as the result below demonstrates, endemic steady states do not exist if $\mathcal{R}_0<1$.

\begin{proposition}
If $\mathcal{R}_0<1$ then model \eqref{F} does not admit a steady state with $H_i^*=\displaystyle\int_0^\infty h_i^*(a)\,\ud a>0$.
\end{proposition}
\begin{proof}
From equations \eqref{040915_3}, \eqref{eq-compartments} we note that  any steady state with $H_i^*>0$ necessarily satisfies the following equations.
\begin{equation}\label{end-steadyeq}
\begin{aligned}
H_i^* & =   \int_0^{\infty}\int_0^a\exp\left\{-\int_{\bar{a}}^a(\mu_{hi}(\eta)+\rho_h(\eta))\,\ud\eta\right\}  \left[\alpha_f\alpha(\bar{a})\,p_{fh}\frac{h_s^*(\bar{a})}{H^*}F_i^*+\varepsilon\,\alpha_{fw}\alpha(\bar{a})\,p_{fh}\frac{h_s^*(\bar{a})}{H^*}F_{wi}^*\right]\,\ud\bar{a}\,\ud a, \\
 \alpha_f\,p_{hf}F_s^*\frac{H_i^*}{H^*} & =  \frac{\mu_f(\mu_f+\nu_f)}{\nu_f}F_i^*,  \\
 \alpha_{fw}\,p_{hf}F_{ws}^*\frac{H_i^*}{H^*} & =  \frac{\mu_{fw}(\mu_{fw}+\nu_{fw})}{\nu_{fw}}F_{wi}^*.
\end{aligned}
\end{equation}
Combining equations \eqref{end-steadyeq} and recalling formula \eqref{160216_1} we conclude that $H_i^*>0$ implies $\mathcal{R}_0=1$, a contradiction.
\end{proof}

The final result we establish characterizes important global properties of our model \eqref{F}. The proof of the result below is contained in Appendix B.

\begin{theorem}\label{theorem2}
Assume that 
\begin{equation}\label{mu-assumption}
\mu_{hi}(a) = \mu_{hs}(a) + \mu_i(a),\quad \mu_{hr}(a) = \mu_{hs}(a), 
\end{equation}
and that there exists some $\mu_{hs}^*>0$, such that
\begin{equation}\label{mu-bounded}
0<\displaystyle\inf_{a\in [0,\infty)}\{\mu_i(a)\}\leq\sup_{a\in [0,\infty)}\{\mu_i(a)\}<\infty,\,\, \mu_{hs}^*\le \displaystyle\inf_{a\in [0,\infty)}\{\mu_{hs}(a)\};\,\,\, \text{and that}  \,\,\, 0<\beta_h(a)\le \hat{\beta},\,\,\, \forall a\in [0,\infty). 
\end{equation}
Moreover, assume that
\begin{equation}\label{con_global}
\int_0^{\infty} \beta_h(a) \exp\left\{-\int_0^a \mu_{hs}(\eta)\,\ud\eta\right\} \ud a= 1.
\end{equation}
Then $\forall\, a\in [0,\infty)$ we have that $h_i(t,a)\to 0$ as $t\to\infty$.
Moreover, we have
\begin{equation}\label{2608164}
\limsup_{t\rightarrow \infty} \int_0^\infty \left(h_s(t,a) + h_i(t,a) + h_r(t,a)\right)\ud a :=\limsup_{t\rightarrow \infty} H_{\rm total}(t)  \leq  \frac{ H_{total}(0)\hat{\beta} \displaystyle\int_0^{\infty} 
  \exp\left\{ - \int_0^a \mu_{hs}(\eta) \,\ud\eta \right\} \ud a  }{\mu_{hs}^* \displaystyle\int_0^{\infty} a \beta_h(a) \exp\left\{ -\int_0^a \mu_{hs}(\eta) \, \ud\eta \right\} \,\ud a}. 
\end{equation}
\end{theorem}

To interpret the results above first note that condition \eqref{con_global} says that the so-called net reproduction number of the susceptible human population is $1$. The net reproduction number is the average of the number of offspring individuals produce during their lifetime.  This quantity (albeit note the obvious similarity) is not to be confused with the basic reproductive number, which characterizes the number of secondary cases a single infected individual causes in a completely susceptible population. Condition \eqref{con_global} simply implies that the susceptible human population is at a steady state; and this is a reasonable assumption, as human population dynamics and dengue disease dynamics take place on different time-scales. We also naturally assumed in \eqref{mu-assumption} that susceptible and recovered human mortality is the same, while infected human mortality is larger than that of susceptibles. The boundedness conditions on the fertility and mortality rates in \eqref{mu-bounded} are technical and they are realistic.

Under these hypotheses \eqref{2608164} establishes a bound for the size of the total human population, and since \eqref{con_global} implies that the susceptible population size does not change (in particular does not grow) due to population dynamics, the upper bound on the total human population size in turn gives an upper bound for the infected human population size; while note that our assumptions do not imply that $\mathcal{R}_0$ is small.

\section{Discussion}

In this work we introduced a coupled system of ordinary and partial differential equations to model dengue infection dynamics in an age-structured human population. To take into account the potential effects of a dengue inhibiting \Wolb strain we used a detailed sex-structured mosquito population model, which accounts for most of the known mechanisms by which \Wolb interferes with its host's reproductive mechanisms. Moreover, our mosquito population model allows to take into account the possible life shortening effect of the dengue inhibiting \Wolb strain, by allowing for a reduction of the reproductive output of \Wolb infected females. The age-structure of the human population clearly impacts dengue infection dynamics for a number of reasons: 
\begin{enumerate}
\item Recovery and death rates of dengue infected humans vary by age. 
\item Older people have more likely to have acquired immunity against specific serotypes of dengue. 
\item Biting rates of mosquitoes also  vary by human age. 
\end{enumerate}

Our  model \eqref{F} allows to study all of these effects and their combined impact on dengue disease dynamics. Indeed, the analytic formula we deduced here for $\mathcal{R}_0$ demonstrates the intricate relationship between a dengue inhibiting \Wolb infection in the mosquito population, and the demographic profile of the human population. In particular note that (as expected) $\mathcal{R}_0$ is a monotone decreasing function of the parameter $\varepsilon$, which measures the strength of the inhibition of dengue replication in \Wolb infected mosquitoes. This would imply in principle that \Wolb induced dengue inhibition in mosquitoes may indeed reduce the severity of dengue outbreaks. On the other hand, the formula we deduced for $\mathcal{R}_0$ also shows that if a new strain of dengue inhibiting \Wolb perturbs the mosquito population into another steady state (note that our model can have multiple co-existence steady states in some parameter regimes), in which an existing/resident strain will dominate (corresponding to the \Wolb uninfected compartment in our model), that is a steady state for which $F_s^*\gg F_{ws}^*$ holds, than this would actually limit the effectiveness of a dengue inhibiting \Wolb release on reducing $\mathcal{R}_0$. Hence studying the possible  coexistence between dengue inhibiting \Wolb strains and resident \Wolb strains in specific geographic locations will be key. Indeed, there is a lot of recent work, see e.g. \cite{Nadin,Turelli, Zhang, Zheng} on modelling the spatial spread of mosquitoes following the  release of a new \Wolb strain. It would be indeed interesting to expand our sex-structured mosquito population model into a spatially structured one, to investigate whether co-existence and  bi-stability remains a possibility in some parameter regimes.

We also note that the formula we deduced for $\mathcal{R}_0$ shows the impact of the acquisition of age-related immunity. In particular we can see from formula \eqref{160216_1} that a decrease of $p_{fh}\,\alpha(\bar{a})$ for large values of $\bar{a}$, corresponding to increased probability of dengue immunity in the elderly, has a negligible effect on the value of $\mathcal{R}_0$, in comparison to a similar decrease of  $p_{fh}\,\alpha(\bar{a})$ for smaller values of $\bar{a}$. This then also implies that future vaccination efforts should be focused primarily on the younger age-groups to maximise impact; and at the same time underlines why dengue continues to be a major challenge in populations with specific demographics. 

Finally we note that the rigorous analysis of infinite dimensional dynamical systems such as \eqref{F} is challenging for a number reasons. We only highlight here one issue, that is the  transitioning between compartments, which is typical for example in models of infectious diseases or cell populations. To obtain qualitative results even for linear or finite dimensional models is non-trivial due to the fact that transitioning between compartments makes it more difficult to exploit positivity (or monotonicity) properties of such systems, see e.g. \cite{FHin,Hadeler2008,Smith95}.

\appendix

\section{Proof of Theorem 3.1}

The characteristic equation associated with the
linearisation about a dengue-free equilibrium is given by \eqref{040915_5}. We are going to prove that all of its roots satisfy ${\rm Re}(\lambda)<0$. Let us
restrict our attention first to possible real roots. Note that
$\Phi(\lambda)$ is a decreasing function of $\lambda$ and therefore
each of the two bracketed terms in the left hand side
of~(\ref{040915_5}) is increasing in $\lambda$, at least for
$\lambda\geq 0$. The right hand side of~(\ref{040915_5}) decreases in
$\lambda$, at least for $\lambda\geq 0$. It follows that if the left
hand side exceeds the right hand side when $\lambda=0$, then all the
real roots of the characteristic equation \eqref{040915_5} are
strictly negative. Using the
expression for $h_s^*(a)$ in~(\ref{040915_1}), and
expression~(\ref{040915_4}), this criterion can be rewritten in the
form $\mathcal{R}_0<1$, where $\mathcal{R}_0$ is given by~(\ref{160216_1}). 

Next we deal with possible complex roots. More precisely, we show that if
the dominant real root of the characteristic equation~(\ref{040915_5})
is negative, then any complex root must have negative real part, too. To see
this, note that the characteristic equation \eqref{040915_5} can be rewritten in the form
\begin{equation}\label{160216_2}
1 =\varepsilon\,\alpha_{fw}^2p_{hf}p_{fh}\frac{F_{ws}^*}{H^{*^2}}\nu_{fw}\frac{\Phi(\lambda)}{(\lambda+\mu_{fw})(\lambda+\mu_{fw}+\nu_{fw})}
 +\alpha_{f}^2p_{hf}p_{fh}\frac{F_{s}^*}{H^{*^2}}\nu_{f}\frac{\Phi(\lambda)}{(\lambda+\mu_{f})(\lambda+\mu_{f}+\nu_{f})}
=:\Psi(\lambda).
\end{equation}
Note that $\Psi(\lambda)$ is decreasing as a function of $\lambda$, for $\lambda$
real and positive. By hypothesis, $\mathcal{R}_0<1$ so that the dominant 
real root of \eqref{040915_5} is negative. We show that under these
circumstances there are no complex roots of \eqref{040915_5} with
${\rm Re}(\lambda)\geq 0$. To see this, note that $\mathcal{R}_0<1$ if and only
if $\Psi(0)<1$. From the definition of $\Phi(\lambda)$ (see
expression \eqref{040915_4}), it is easily seen that
$|\Phi(\lambda)|\leq \Phi({\rm Re}(\lambda))$. Furthermore, 
if ${\rm Re}(\lambda)\geq 0$ then $|\lambda+\mu_f|\geq {\rm
  Re}(\lambda)+\mu_f$, with similar estimates for the other factors in
the denominators of expression \eqref{160216_2}. Therefore, any complex
root $\lambda$ of the characteristic equation \eqref{040915_5} with 
${\rm Re}(\lambda)\geq 0$ must satisfy
\begin{equation*}
\begin{aligned}
1 =|\Psi(\lambda)|\leq &
\,\,\varepsilon\,\alpha_{fw}^2p_{hf}p_{fh}\frac{F_{ws}^*}{H^{*^2}}\nu_{fw}\frac{\Phi({\rm
    Re}(\lambda))}{({\rm Re}(\lambda)+\mu_{fw})({\rm Re}(\lambda)+\mu_{fw}+\nu_{fw})} \\
&  +\alpha_{f}^2p_{hf}p_{fh}\frac{F_{s}^*}{H^{*^2}}\nu_{f}\frac{\Phi({\rm
     Re}(\lambda))}{({\rm Re}\,\lambda+\mu_{f})({\rm Re}(\lambda)+\mu_{f}+\nu_{f})}
:=\Psi({\rm Re}(\lambda)).
\end{aligned}
\end{equation*}
Since $\Psi(\lambda)$, restricted to real $\lambda$, is decreasing for $\lambda\geq 0$, and since $\Psi(0)<1$ and $\Psi({\rm Re}(\lambda))\geq 1$ (see above), it follows that ${\rm
  Re}(\lambda)<0$, giving a contradiction. \qed

\section{Proof of Theorem 3.3}

We introduce $h_{\rm{total}}(t,a) := h_s(t,a) + h_i(t,a) + h_r(t,a)$ and note from \eqref{F} that it satisfies the boundary value problem:
\begin{equation}\label{2508161}
\begin{aligned}
\frac{\partial h_{\rm{total}}(t,a)}{\partial t} + \frac{\partial h_{\rm{total}}(t,a)}{\partial a}    & = - \mu_{hs}(a)h_{\rm{total}}(t,a) - \mu_i(a) h_i(t,a),  \\
 h_{\rm{total}}(t,0) &= \int_0^{\infty}  \beta_h(a) (h_s(t,a) + h_r(t,a)) \, \ud a .
\end{aligned}
\end{equation}
Therefore we have
\begin{equation}\label{2508162}
h_{\rm{total}}(t,0) \leq \int_0^{\infty}  \beta_h(a) h_{\rm{total}}(t,a) \, \ud a = \int_0^t \beta_h(a) h_{\rm{total}}(t,a) \, \ud a + \int_t^{\infty}  \beta_h(a) h_{\rm{total}}(t,a) \, \ud a.
\end{equation}
Introducing 
\begin{equation*}
N(t):= h_{\rm{total}}(t,0), \quad N_0(t):=\int_t^{\infty}  \beta_h(a) h_{\rm{total}}(t,a) \, \ud a,
\end{equation*}
inequality \eqref{2508162} can be written in the form 
\begin{equation}\label{2508163}
N(t) \leq \int_0^t \beta_h(a) h_{\rm{total}}(t,a) \, \ud a + N_0(t).
\end{equation}
Note that for large times the first term in \eqref{2508163} dominates.  The second term $N_0(t)$ in \eqref{2508163} depends on the initial age distribution $h_{\rm{total}}(0,a-t)$ and loses its influence as $t\rightarrow \infty$.

Next we deduce an expression for $h_{\rm{total}}(t,a)$ valid for $t \geq a$. Introducing
$$
h^{\xi}_{\rm{total}}(a) = h_{\rm{total}}(a + \xi,a),
$$
the first equation in \eqref{2508161} can be written in the form
$$
\frac{d h^{\xi}_{\rm{total}}(a)}{d a} + \mu_{hs}(a) h^{\xi}_{\rm{total}}(a) = - \mu_i(a) h_i(a + \xi, a).
$$
Solving this ODE for $h^{\xi}_{\rm{total}}(a)$ and then putting $\xi = t-a$, gives 
\begin{equation}
h_{\rm{total}}(t,a) =  \, h_{\rm{total}}(t-a,0) \exp\left\{-\int_0^a \mu_{hs}(\eta) \, \ud\eta \right\}- \int_0^a \mu_i(\xi) \exp\left\{ -\int_{\xi}^a  \mu_{hs}(\eta) \, \ud\eta \right\} h_i(\xi+t-a,\xi) \, \ud\xi,
 \end{equation}
for $t>a$. Replacing $h_{\rm{total}}(t-a,0)$ by $N(t-a)$, inequality \eqref{2508163} becomes 
\begin{equation}\label{2508164}
N(t) \leq \int_0^t  f(a)N(t-a)   \,da -G(t) + N_0(t),
\end{equation}
where 
\begin{equation}\label{2508165}
f(a) = \beta_h(a) \exp\left\{- \int_0^a \mu_{hs} (\eta) \, \ud\eta \right\},
\end{equation}
and $G(t)$ is defined as
\begin{equation}\label{Gt}
G(t) = \int_0^t \beta_h(a) \int_0^a \mu_i(\xi) \exp\left\{ -\int_{\xi}^a \mu_{hs} (\eta) \,\ud\eta \right\} h_i(\xi + t -a, \xi) \,\ud \xi \,\ud a. 
\end{equation}

By a comparison argument, using that $f(a) \geq 0$, we may say that $N(t) \leq \bar{N}(t)$, where $\bar{N}(t)$ satisfies the integral equation 
\begin{equation}\label{2508166}
\bar{N}(t) = \int_0^t  f(a)\bar{N}(t-a)   \,\ud a -G(t) + N_0(t).
\end{equation}
Applying the Laplace transform, and using the convolution theorem we have 
\begin{equation}\label{laplace}
\hat{N}(s) = \hat{f}(s) \hat{N}(s) + \hat{N}_0(s) - \hat{G}(s),
\end{equation}
where $s$ is the transform variable;  $\hat{N}(s)$ is the Laplace transform of $\bar{N}(t)$, $\hat{G}(s)$ is the Laplace transform of $G(t)$, and $\hat{N}_0(s) $ is the Laplace transform of $N_0(t)$. Note that $\hat{N}_0(s)$ satisfies 
\begin{equation}\label{laplace-est}
\hat{N}_0(s)\le \int_0^\infty e^{-st}\int_t^\infty \beta_h(a)h_{\rm{total}}(0,a-t)\exp\left\{-\int_{a-t}^a\mu_{hs}^*\,dx\right\}\, \ud a\, \ud t\le \frac{\hat{\beta}H_{\rm{total}}(0)}{\mu_{hs}^*+s}, \quad s>-\mu_{hs}^*.
\end{equation} 
Solving \eqref{laplace} for $\hat{N}(s)$, then applying the inversion formula for the Laplace transform yields 
\begin{equation}\label{2508167}
\bar{N}(t) = \frac{1}{2\pi i} \int_{\sigma - i \infty}^{\sigma + i \infty} e^{st} \frac{\hat{N}_0(s) - \hat{G}(s)}{1 - \hat{f}(s)} \,\ud s
\end{equation}
where $\sigma$ is any real number that strictly exceeds the real parts of all the poles of integrand.
These poles  are the zeros of 
\begin{equation*}
1 = \hat{f}(s) = \int_0^{\infty} e^{-sa} \beta_h(a) \exp\left\{ - \int_0^a \mu_{hs}(\eta)\,\ud\eta \right\} \,\ud a.
\end{equation*}
By  assumption \eqref{con_global},  the poles are $s=0$ and (possibly) infinitely many other poles all of which satisfy ${\rm Re}(s) < 0$. The residues at these other poles all decay exponentially with $t$, so we focus on the pole at $s=0$ and compute its residue using a standard formula. In particular, by Cauchy's residue theorem, for large $t$,
\begin{equation*}
\bar{N}(t) \simeq {\rm Res} \left[\frac{e^{st} ( \hat{N}_0(s) - \hat{G}(s))}{1-\hat{f}(s)},\, s=0 \right] = \frac{\hat{N}_0(0) - \hat{G}(0)}{-\hat{f}^{\prime}(0)}.
\end{equation*}
But 
\begin{equation*}
\hat{f}^{\prime}(s) = -\int_0^{\infty} a e^{-sa} \beta_h(a) \exp\left\{ -\int_0^a \mu_{hs}(\eta) \, \ud\eta \right\} \, \ud a,
\end{equation*}
so we have
\begin{equation*}
\hat{f}^{\prime}(0) = -\int_0^{\infty} a \beta_h(a) \exp\left\{ -\int_0^a \mu_{hs}(\eta) \, \ud\eta \right\} \, \ud a.
\end{equation*}
Since $N(t) \leq \bar{N}(t)$, 
\begin{equation}\label{2608163b}
\limsup_{t\rightarrow \infty} N(t) \leq \lim_{t\rightarrow \infty} \bar{N}(t) = \frac{\hat{N}_0(0) - \hat{G}(0)}{\displaystyle\int_0^{\infty} a \beta_h(a) \exp\left\{ -\int_0^a \mu_{hs}(\eta) \, \ud\eta \right\} \,\ud a}.
\end{equation}
Note that  $\displaystyle \hat{G}(0) = \int_0^{\infty} G(t)\, \ud t$, so the only way to have $\displaystyle \limsup_{t\rightarrow \infty} N(t) \geq 0$, is to have 
\begin{equation*}
\hat{G}(0) = \int_0^{\infty} G(t)\,\ud t < \infty.
\end{equation*}
Thus, $G(t)\rightarrow 0$ as $t\rightarrow\infty$. To establish the first assertion of Theorem 3.3 we rewrite $G(t)$ as
\begin{equation}\label{Gt-2}
G(t)=\int_0^t \beta_h(a)\exp\left\{-\int_0^a\mu_{hs}(\eta)\,\ud \eta\right\} \int_0^a \mu_i(\xi) \exp\left\{ \int_0^{\xi} \mu_{hs} (\eta) \,\ud\eta \right\} h_i(\xi + t -a, \xi) \,\ud \xi \,\ud a. 
\end{equation}
Because of our assumption \eqref{con_global} and since $\beta_h>0$ from \eqref{mu-bounded}, $G(t)\to 0$ as $t\to\infty$ implies that 
$$\int_0^a \mu_i(\xi) \exp\left\{ \int_0^{\xi} \mu_{hs} (\eta) \,\ud\eta \right\} h_i(\xi + t -a, \xi) \,\ud \xi\to 0,\quad\forall\,a\in [0,\infty),$$
and since the integrand is positive and the integral as a function of $a$ is continuous, we have that $h_i(\xi + t -a, \xi)\to 0,\,\,\forall\,\xi,\,a$ as $t\to \infty$, hence taking $\xi=a$ yields $h_i(t,a)\to 0$ as $t\to\infty$.

Next note that $G(t)$, defined in \eqref{Gt}, satisfies

\begin{equation*}
\begin{aligned}
G(t)   \geq & \int_0^t  \int_0^a \mu_i(\xi) \beta_h(a) \exp\left\{ -\int_0^a \mu_{hs} (\eta) \,\ud\eta \right\} h_i(\xi + t -a, \xi) \,\ud \xi \,\ud a\\
  = &  \int_0^t \int_0^a \mu_i(\xi) f(a) h_i(\xi+t-a,\xi) \,\ud\xi \, \ud a \\
\geq & \inf_{\xi\in [0,\infty)}\{ \mu_i(\xi)\} \int_0^t f(a) \int_0^a h_i(\xi+t-a,\xi) \,\ud\xi \,\ud a, 
\end{aligned}
\end{equation*}
so that $G(t)\geq \displaystyle\inf_{\xi\in [0,\infty)}\{ \mu_i(\xi)\}\, \tilde{G}(t)$, where
\begin{equation*}
\tilde{G}(t)= \int_0^t f(a) \int_0^a h_i(\xi+t-a,\xi) \,\ud\xi \,\ud a.
\end{equation*}
Thus $\tilde{G}(t) \rightarrow 0$ as $t\rightarrow\infty$. Next note that, for any finite $T>0$,
\begin{equation*}
\tilde{G}_T(t):= \int_0^T f(a) \int_0^a h_i(\xi+t-a,\xi) \,\ud\xi \,\ud a\leq \tilde{G}(t),
\end{equation*}
for all $t\geq T$, and thus we also have $\tilde{G}_T(t) \rightarrow 0$ as $t\rightarrow\infty$. Thus
\begin{equation*}
0=\lim_{t\rightarrow \infty} \tilde{G}_T(t) = \int_0^T f(a)  \lim_{t\rightarrow \infty} \int_0^a h_i(\xi+t-a,\xi) \,\ud\xi \,\ud a,
\end{equation*}
by the dominated convergence theorem. Since $f(a)>0$ for all $a\geq 0$, we have
\begin{equation}\label{2608162}
\lim_{t\rightarrow \infty} \int_0^a h_i(\xi+t-a,\xi) \,\ud\xi = 0,
\end{equation}
for each fixed $a$. Note that 
\begin{equation}\label{2608161}
\begin{aligned}
H_{total}(t)  = & \int_0^\infty h_{\rm{total}}(t,a)\, \ud a=\int_0^\infty \left(h_s(t,a) + h_i(t,a) + h_r(t,a)\right)\ud a \\
= &  \int_0^{\infty} N(t-a) \exp\left\{ - \int_0^a \mu_{hs}(\eta) \,d\eta \right\} \ud a
 - \int_0^{\infty} \int_0^a \mu_i(\xi)  \exp\left\{ - \int_{\xi}^a \mu_{hs}(\eta) \,\ud\eta \right\} h_i(\xi+t-a,\xi) \, \ud\xi \, \ud a.
 \end{aligned}
\end{equation}
The last term in \eqref{2608161} is bounded by
\begin{equation*}
\displaystyle\sup_{\xi\in [0,\infty)}\{\mu_i(\xi)\} \int_0^{\infty} \int_0^a   h_i(\xi+t-a,\xi) \, \ud\xi \, \ud a,
\end{equation*}
which, by the dominated convergence theorem, tends to zero as $t\rightarrow\infty$, by \eqref{2608162}. Therefore, for any given $\epsilon>0$ the inequality 
\begin{equation*}
\int_0^{\infty}  N(t-a)  \exp\left\{  - \int_0^a \mu_{hs}(\eta) \,\ud\eta \right\} \ud a - \epsilon 
 \leq H_{total}(t) \leq \int_0^{\infty}  N(t-a)  \exp\left\{ - \int_0^a \mu_{hs}(\eta) \,\ud\eta \right\} \ud a
\end{equation*}
holds, for sufficiently large $t$. A consequence of this inequality is that, using \eqref{2608163b} and Fatou's lemma,
\begin{equation*}
\begin{aligned}
\limsup_{t\rightarrow \infty} H_{total}(t)  & \leq  \int_0^{\infty} \left[\limsup_{t\rightarrow \infty} N(t-a)\right]  \exp\left\{ - \int_0^a \mu_{hs}(\eta) \,\ud\eta \right\} \ud a \\
 & \leq  \frac{\hat{N}_0(0) }{\displaystyle\int_0^{\infty} a \beta_h(a) \exp\left\{ -\int_0^a \mu_{hs}(\eta) \, \ud\eta \right\} \, \ud a} \displaystyle\int_0^{\infty} 
  \exp\left\{ - \int_0^a \mu_{hs}(\eta) \,\ud\eta \right\} \ud a 
  \end{aligned}
\end{equation*}
which, combined with \eqref{laplace-est}, establishes \eqref{2608164}. \qed

\section*{Acknowledgments}
We thank the American Institute of Mathematics for financial support through a SQuaRE program during 2014-2016.


\begin{thebibliography}{99}

\bibitem{Aguas}
Aguas, R., Dorigatti, I., Coudeville, L. et al., 
\newblock Cross-serotype interactions and disease outcome prediction of dengue infections in Vietnam. \newblock {\it Sci. Rep.} {\bf 9}, 9395 (2019). 

\bibitem{Bhatt}
Bhatt, S., Gething, P., Brady, O. et al., 
\newblock The global distribution and burden of dengue,
\newblock {\it  Nature} {\bf 496}, 504-507 (2013).

\bibitem{Blagrove}
Blagrove, M.S.C. et al.,
\newblock {\it Wolbachia} strain {\it wMel} induces cytoplasmic incompatibility and blocks dengue transmission in {\it Aedes albopictus},
\newblock {\it PNAS}, {\bf 109} (2012), 255-260.

\bibitem{BullTurelli}
Bull, J.J., Turelli, M.,
\newblock Wolbachia versus dengue Evolutionary forecasts,
\newblock {\it Evolution, Medicine, and Public Health},  (2013) pp. 197-207.

\bibitem{Dorigatti}
Dorigatti, I., McCormack, C., Nedjati-Gilani, G., Ferguson, N.M.,
\newblock Using {\it Wolbachia} for Dengue Control: Insights from Modelling,
\newblock {\it Trends in Parasitology},  {\bf 34} , No. 2. (2008) 102-113.

\bibitem{Engelstadter2006}
Engelsta\"{a}dter, J. Charlat, S., Pomiankowski, A., Hurst, G.D.D.,
\newblock The Evolution of Cytoplasmic Incompatibility Types: IntegratingSegregation, Inbreeding and Outbreeding,
\newblock {\it Genetics}  (2006) {\bf 172}(4): 2601-2611.

\bibitem{Engelstadter2009}
Engelst\"adter, J., Telschow, A., 
\newblock Cytoplasmic incompatibility and host population structure,
\newblock {\it Heredity}, \textbf{103} (2009), 196-207.

\bibitem{Engelstadter2004}
J. Engelst\"adter, A. Telschow, and P. Hammerstein,
\newblock Infection dynamics of different \textit{Wolbachia}-types within one host population,
\newblock {\it J. Theor. Biol.}, \textbf{231} (2004), 345-355.

\bibitem{FGLY}
Farkas, J.Z., Gourley, S.A., Liu, R.,  Yakubu, A.-A.,
\newblock Modelling {\it Wolbachia} infection in a sex-structured mosquito population carrying West Nile virus, 
\newblock {\it J. Math. Biol.}  {\bf 75} (2017), 621-647. 

\bibitem{FHin}
Farkas, J.Z., Hinow, P., 
\newblock On a size-structured two-phase population model with infinite states-at-birth, 
\newblock {\it Positivity}, {\bf 14} (2010), 501-514.

\bibitem{FHin2}
Farkas, J.Z., Hinow, P., 
\newblock Structured and unstructured continuous models for \textit{Wolbachia} infections,
\newblock {\it Bull. Math. Biol.}, \textbf{72} (2010), 2067-2088.

\bibitem{Hadeler2008}
Hadeler, K.P., Thieme, H.R.,
\newblock Monotone dependence of the spectral bound on the transition rates in linear compartment models, 
\newblock {\it J. Math. Biol.}  {\bf 57} (2008), 697-712.

\bibitem{Hancock}
Hancock, P.A., Sinkins, S.P., Godfray, H.C.J.,
\newblock Population dynamic models of the spread of {\it Wolbachia},
\newblock {\it Am. Naturalis} {\bf 177} (2011), 323-333.

\bibitem{Hancock2}
Hancock, P.A., Sinkins, S.P., Godfray, H.C.J.,
\newblock Strategies for introducing {\it Wolbachia} to reduce transmission of mosquito-borne diseases,
\newblock {\it PLoS Negl. Trop. Dis.}, {\bf 5} (2011), e1024.

\bibitem{Hilgen}
Hilgenboecker, K., H., Hammerstein, P., Schlattmann, P., Telschow, A., Werren, J.H.,
\newblock How many species are infected with Wolbachia? - a statistical analysis of current data,
\newblock {\it FEMS Microbiology Letters}, {\bf 281} (2008), 215-220.

\bibitem{Hoffmann1997}
Hoffmann, A.A., Turelli, M.,
\newblock Cytoplasmic incompatibility in insects,
\newblock In A. A. Hoffmann S. L. O'Neill and J. H. Werren (eds), 
\newblock {\it Influential Passengers} (1997),  Oxford University Press, 42-80.

\bibitem{Hoffmann}
Hoffmann, A.A., et al.,
\newblock Successful establishment of {\it Wolbachia} in {\it Aedes} populations to suppress dengue transmission,
\newblock {\it Nature},  {\bf 476} (2011), 454-457.

\bibitem{Hughes_2013}
Hughes, H., Britton, N.F., 
\newblock Modelling the use of Wolbachia to control dengue fever transmission,
\newblock {\it Bull. Math. Biol.}, {\bf 75} (2013), 796-818.

\bibitem{Hurst}
Hurst, G.D.D., Jiggins, F.M., Pomiankowski, A.,
\newblock Which Way to Manipulate Host Reproduction? Wolbachia That Cause Cytoplasmic Incompatibility Are Easily Invaded by Sex Ratio–Distorting Mutants,
\newblock {\it The American Naturalist}, {\bf 160} (2012), 360-373.

\bibitem{Keeling2003}
Keeling, M.J., Jiggins, F.M., Read, J.M.,
\newblock The invasion and coexistence of competing \textit{Wolbachia} strains,
\newblock {\it Heredity}, \textbf{91} (2003), 382-388.

\bibitem{PengLu}
Lu, P.,  Bian G., Pan X., Xi Z., 
\newblock {\it Wolbachia} Induces Density-Dependent Inhibition to Dengue Virus in Mosquito Cells, 
\newblock {\it PLoS Negl Trop Dis} {\bf 6(7)}: e1754 (2012).


\bibitem{Maier}
Maier, S.B., Massad, E., Amaku, M., Burattini, M.N. Greenhalgh, D.,
\newblock The Optimal Age of Vaccination Against Dengue with an Age-Dependent Biting Rate with Application to Brazil,
\newblock {\it Bull Math Biol.} (2020)  {\bf 82}(1): 12

\bibitem{McMeniman2009}
McMeniman, C.J., et al.,
\newblock Stable introduction of a life-shortening \textit{Wolbachia} infection into the mosquito \textit{Aedes aegypti},
\newblock {\em Science}, \textbf{323} (2009), 141-144.

\bibitem{Ndii2016}
Ndii, M.Z., Allingham, D., Hickson, R.I., Glass, K.,
\newblock The effect of {\it Wolbachia} on dengue outbreaks when dengue is
repeatedly introduced,
\newblock {\em Theoretical Population Biology}  {\bf 111} (2016) 9-15.

\bibitem{Ndii2015}
Ndii, M.Z., Hickson, R.I., Allingham, D., Mercer, G.N.,
\newblock Modelling the transmission dynamics of dengue in the presence of {\it Wolbachia},
\newblock {\em Mathematical Biosciences}  {\bf 262} (2015) 157-166.

\bibitem{Nadin}
Nadin, N., Strugarek, M., Vauchelet, N., 
\newblock Hindrances to bistable front propagation: application to Wolbachia invasion,
\newblock {\em Journal of Mathematical Biology} {\bf 76},  (2018), 1489-1533.

\bibitem{Nagumo1942}
Nagumo, M.
\newblock \"{U}ber die Lage der Integralkurven gew\"{o}hnlicher Differentialgleichungen.
\newblock {\em Proc. Phys.-Math. Soc. Japan} {\bf 24} (1942), 551-559.


\bibitem{ONeill}
O'Neill, S.L., Hoffmann, A.A.,  Werren, J.H., (eds),
\newblock {\em Influential Passengers}.
\newblock Oxford University Press, Oxford, New York, Tokyo, (1997).

\bibitem{Pruss2010}
Pr\"{u}ss, J.W., Wilke, M.,
\newblock {\em Gew\"{o}hnliche Differentialgleichungen und dynamische Systeme.}
\newblock Grundstudium Mathematik. Birkh\"{a}user/Springer Basel AG, Basel, (2010).

\bibitem{Rasgon2004}
Rasgon, J.L., Scott, T.W.,
\newblock Impact of population age structure on \textit{Wolbachia} transgene driver efficacy: ecologically complex factors and release of genetically modified mosquitoes,
\newblock {\em Insect Biochem. Mol. Biol.}, \textbf{34} (2004), 707-713.

\bibitem{Ritchie}
Ritchie, S.A., Pyke, A.T., Hall-Mendelin, S., Day, A., Mores, C.N., et al., 
\newblock An Explosive Epidemic of DENV-3 in Cairns, Australia
\newblock {\em PLoS ONE} {\bf 8}(7): e68137.

\bibitem{Rock}
Rock, K.S., Wood, D.A., Keeling, M.J.,
\newblock Age- and bite-structured models for vector-borne diseases,
\newblock {\em Epidemics}, {\bf 12} (2015), 12-20.

\bibitem{Rybski}
Rybski D., et al.
Scaling laws of human interaction activity.
{\em PNAS} August 4, (2009) {\bf 106} (31) 12640-12645; https://doi.org/10.1073/pnas.0902667106


\bibitem{Smith95} 
Smith, H.L.,
\newblock {\em Monotone dynamical systems. An introduction to the theory of competitive and cooperative systems}. 
\newblock Mathematical Surveys and Monographs, 41. American Mathematical Society, Providence, RI, (1995).

\bibitem{Stouthamer1997}
Stouthamer, R.
\newblock \textit{Wolbachia} induced parthenogenesis,
\newblock In A. A. Hoffmann S. L. O'Neill and J. H. Werren, editors, {\em Influential Passengers} (1997), Oxford University Press, 102-124.

\bibitem{Telschow2005}
Telschow, A., Hammerstein, P., Werren, J.H.,
\newblock The effect of \textit{Wolbachia} versus genetic incompatibilities on reinforcement and speciation,
\newblock {\em Evolution}, \textbf{59} (2005), 1607-1619.

\bibitem{Telschow2005b}
Telschow, A., Yamamura, N.,  Werren, J.H.,
\newblock Bidirectional cytoplasmic incompatibility and the stable coexistence of two \textit{Wolbachia} strains in parapatric host populations,
\newblock {\em J. Theor. Biol.}, \textbf{235} (2005), 265-274.


\bibitem{Thai}
Thai, K.T.D., Nishiura, H., Hoang, P.L., Tran, N.T.T., Phan, G.T., et al., 
\newblock Age-Specificity of Clinical Dengue during Primary and Secondary Infections.
\newblock {\em PLoS Negl. Trop. Dis.} (2011) {\bf 5}(6): e1180.

\bibitem{Turelli}
Turelli, M., Barton, N.H.,
\newblock Deploying dengue-suppressing {\it Wolbachia}: Robust models predict slow but effective spatial spread in {\it Aedes aegypti}.
\newblock {\em Theoretical Population Biology}, {\bf 115}, (2017) 45-60.

\bibitem{Turelli1994}
Turelli,  M.,
\newblock Evolution of incompatibility-inducing microbes and their hosts, 
\newblock {\em Evolution}, {\bf 48} (1994), 1500-1513.

\bibitem{Vautrin2007}
Vautrin, E., et al.,
\newblock Evolution and invasion dynamics of multiple infections with \textit{Wolbachia } investigated using matrix based models,
\newblock {\em J. Theor. Biol.}, \textbf{245} (2007), 197-209.

\bibitem{Walker2011}
Walker, T., et al.,
\newblock The wMel {\it Wolbachia} strain blocks dengue and invades caged {\it Aedes aegypti} populations,
\newblock {\em Nature},  {\bf 476} (2011), 450-453.

\bibitem{Werren1997}
Werren, J.H.,
\newblock Biology of \textit{Wolbachia},
\newblock {\em Annu. Rev. Entomol.}, \textbf{42} (1997), 587-609.

\bibitem{WHO}
World Health Organization,
\newblock Dengue vaccine: WHO position paper,  {\em Weekly Epidemiological Record}. {\bf 36} (2018)  (93): 457-476.

\bibitem{Zhang}
Zhang, H.,  Lui, R.,
\newblock Releasing {\em Wolbachia}-infected {\em Aedes aegypti} to prevent the spread of dengue virus: A mathematical study,
\newblock {\em Infectious Disease Modelling} {\bf 5} (2020) 142-160.

\bibitem{Zheng}
Zheng, B., Moxun, T., Jianshe, Y.,
\newblock Modeling Wolbachia spread in mosquitoes through delay differential equations, 
\newblock {\em SIAM J. Appl. Math.} {\bf  74}  (2014),  743-770. 


\end{thebibliography}
\end{document}